\documentclass[conference]{IEEEtran}
\IEEEoverridecommandlockouts
\usepackage{multicol}
\usepackage{cite}
\usepackage{amsmath}
\usepackage{amsfonts}
\usepackage{pifont}
\usepackage{amssymb}
\usepackage{amsthm}
\usepackage{tikz}
\usepackage{cases}
\usepackage{graphicx}
\usepackage{float,stfloats}
    \graphicspath{{../}}
    \DeclareGraphicsExtensions{.pdf}
\usepackage[justification=centering]{caption}
\usepackage[caption=false,font=footnotesize]{subfig}
\usepackage{multirow}
\usepackage{booktabs}
\usepackage{url}
\usepackage{xtab}
\usepackage{tabu}
\usepackage{longtable}
\usepackage{algorithm}
\usepackage{algorithmic}
\usepackage{enumerate}
\usepackage{makecell}
\usepackage{lipsum}
\usepackage{multicol}
\usepackage{mathdots}
\usepackage{extarrows}
\usepackage{color,xcolor}
\usepackage{bm}
\usetikzlibrary{arrows}
\usepackage{caption}
\theoremstyle{plain}

\newtheorem{theorem}{Theorem}

\theoremstyle{definition}

\newtheorem{remark}{Remark}

\def\BibTeX{{\rm B\kern-.05em{\sc i\kern-.025em b}\kern-.08em
    T\kern-.1667em\lower.7ex\hbox{E}\kern-.125emX}}

\begin{document}
\title{Explicit Construction of Minimum Storage Rack-Aware Regenerating Codes for All Parameters \thanks{This work was supported in part by the National Key R\&D Program of China  (No. 2020YFA0712300) and NSFC (No. 61872353).}}

\author{\IEEEauthorblockN{Liyang Zhou, Zhifang Zhang}
\IEEEauthorblockA{\fontsize{9.8}{12}\selectfont KLMM, Academy of Mathematics and Systems Science, Chinese Academy of Sciences, Beijing 100190, China\\
School of Mathematical Sciences, University of Chinese Academy of Sciences, Beijing 100049, China\\
Emails: zhouliyang17@mails.ucas.ac.cn,~~zfz@amss.ac.cn}
}
\maketitle

\thispagestyle{empty}

\begin{abstract}
We consider the rack-aware storage system where \(n\!=\!\bar{n}u\) nodes are organized in \(\bar{n}\) racks each containing \(u\) nodes, and any \(k\!=\!\bar{k}u\!+\!u_0~(0\!\leq\! u_0\!<\!u)\) nodes can retrieve the original data file. More importantly, the cross-rack communication cost is
much more expensive than the intra-rack communication cost, so that the latter is usually neglected in the system bandwidth. The MSRR (minimum storage rack-aware regenerating) code is an important variation of regenerating codes that achieves the optimal repair bandwidth for single node failures in the rack-aware model. However, explicit construction of MSRR codes for all parameters were not developed until Chen\&Barg's work. In this paper
we present another explicit construction of MSRR codes for all parameters that improve Chen\&Barg's construction in two aspects: (1) The sub-packetization is reduced from \((\bar{d}-\bar{k}+1)^{\bar{n}}\) to \((\bar{d}-\bar{k}+1)^{\lceil\frac{\bar{n}}{u-u_{0}}\rceil}\) where $\bar{d}$ is the number of helper racks that participate in the repair process;
(2) The field size is reduced to \(|F|\!>\!n\) which is almost half of the field used in Chen\&Barg's construction.
Besides, our code keeps the same access level as Chen\&Barg's low-access construction.
\end{abstract}

\section{Introduction}\label{sec0}
\IEEEPARstart{I}{n} large-scale distributed storage systems,
node failures occasionally happen. A self-sustaining system should be able to recover the data stored in failed nodes by downloading data from surviving nodes. An important metric of repair efficiency is the repair bandwidth, i.e., the total amount of data transmitted during the repair process. Regenerating codes are a kind of erasure codes used in distributed storage systems that can optimize the repair bandwidth for given storage overhead \cite{Dimakis2011}. Particularly, the ones with the minimum storage, i.e., MSR codes, are
appealing in practice in spite of their intricate constructions \cite{Kumar2011,Sasidharan2015,Rawat2016,Ye2016}. The main reason that MSR codes can achieve the optimal repair bandwidth is dividing the data stored in each node into sub-packets of which only a fraction is downloaded from each helper node for repair. The number of sub-packets stored in each node is termed the {\it sub-packetization}. It has been proved that exponential sub-packetization is necessary for MSR codes \cite{SubPBoundSTOC19}. Since the sub-packetization level is closely related to the implementation complexity of the underlying codes, reducing the sub-packetization is significant in practice. Another metric of repair efficiency is the volume of accessed data at the helper nodes which characterizes the disk I/O cost. MSR codes with both the optimal-access property and near optimal sub-packetization were built in \cite{Ye2016sub-}.

The MSR code applies to a homogeneous distributed storage model where all nodes as well as communication between them  are treated indifferently. However,
modern data centers often have hierarchical topologies by organizing nodes in racks, where the cross-rack communication cost is much more expensive than the intra-rack communication cost. This motivates a number of studies that address the repair problem for hierarchical data centers. In this work, we focus on the rack-aware storage model defined as follows.
\begin{table*}[htbp]
\renewcommand\arraystretch{1.75}
\begin{center}
\begin{tabular}{|c|c|c|c|c|}
	\hline  & sub-packetization $\alpha$& access per rack & $\bar{d}$ &field size $|F|$ \\
		\hline Z. Chen et al. \cite{Chen}& $\bar{s}^{\bar{n}} $& $u\cdot\bar{s}^{\bar{n}-1}$& $\bar{k}\leq\bar{d}\leq\bar{n}-1$ &$n|(|F|-1)$ and $|F|\geq n+\bar{s}-1$\\
	\hline H. Hou et al. \cite{Hou2020}& $\bar{s}^{\lceil\bar{n}/\bar{s}\rceil}$ & $\bar{s}^{\lceil\bar{n}/{\bar{s}}\rceil-1}+(u-1)\cdot\bar{s}^{\lceil\bar{n}/{\bar{s}}\rceil}$& $\bar{d}=\bar{n}-1$ & $|F|>k\alpha\sum_{i=1}^{\min\{k,\bar{n}\}}\binom{n-\bar{n}}{k-i}\binom{\bar{n}}{i}$\\
	\hline
This paper & $\bar{s}^{\lceil\frac{\bar{n}}{u-u_{0}}\rceil}$&$u\cdot\bar{s}^{\lceil\frac{\bar{n}}{u-u_{0}}\rceil-1}$ & $\bar{k}\leq\bar{d}\leq\bar{n}-1$ & $ u|(|F|\!-\!1)$ and $|F|\!>\!n$ \\
  \hline
\end{tabular}
\end{center}
\caption{\scriptsize Comparisons with existing constructions of $(n=\bar{n}u,k=\bar{k}u+u_{0})$ MSRR codes where $\bar{s}=\bar{d}-\bar{k}+1$.}\label{t0}
\end{table*}
Suppose $n=\bar{n}u$ and the $n$ nodes are organized in $\bar{n}$ racks each containing $u$ nodes.
A data file consisting of $B$ symbols is stored across the $n$ nodes each storing $\alpha$ symbols such that any $k=\bar{k}u+u_{0}$ ($0\!\leq\! u_{0}\! < \!u$) nodes can retrieve the data file. To rule out the trivial case, we assume throughout that $k\geq u$ \footnote{When $k<u$, a single node erasure can be trivially recovered by the $u-1$ surviving nodes within the same rack because they are sufficient to retrieve the data file.}. Suppose a node fails. The repair process is to generate a replacement node that stores exactly the data of the failed node. The rack that contains the failed node is called the {\it host rack}. The repair is based on the two kinds of communication below:
\begin{enumerate}
\item \textbf{Intra-rack transmission.}
All surviving nodes in the host rack transmit information to the replacement node.
\item \textbf{Cross-rack transmission.}
Outside the host rack, $\bar{d}$ helper racks each transmit $\beta$ symbols to the replacement node.
\end{enumerate}

\noindent Since the cost of intra-rack communication is negligible compared with that of the cross-rack communication, the nodes within each rack can communicate
freely without taxing the system bandwidth. Consequently, the $\beta$ symbols provided by each helper rack are computed from the data stored in all nodes in that helper rack, and the repair bandwidth $\gamma$ only dependents on the cross-rack transmission, i.e., $\gamma=\bar{d}\beta$.

This rack-aware storage model was introduced in \cite{Hu}\cite{Hou}. Moreover, the authors of \cite{Hou} derived a tradeoff between the repair bandwidth and storage overhead for $\bar{k}\leq \bar{d}\leq \bar{n}-1$. The codes with parameters lying on the tradeoff curve are called rack-aware regenerating codes. In particular, the minimum storage rack-aware regenerating (MSRR) code has parameters:
\begin{equation}\label{bw-bound}
\alpha={B}/{k},\ \ \ \beta=\alpha/(\bar{d}-\bar{k}+1)\;.
\end{equation}
Certainly $B,\alpha,\beta$ are all integers and $\alpha$ is called the sub-packetization. On the one hand, codes with small sub-packetization are preferred in practice due to the low complexity in both the encoding and repair process. On the other hand,
$\alpha$ must be large enough to guarantee the existence of MSRR codes for arbitrary $n,k$. It was proved in \cite{Chen} that  $(n\!=\!\bar{n}u,k\!=\!\bar{k}u,\bar{k}\!\leq\! \bar{d}\!\leq\! \bar{n}\!-\!1)$ optimal-access (i.e., the symbols accessed on each helper rack are downloaded without processing) MSRR codes exist only if $\alpha \geq \min\{\bar{s}^{\frac{\bar{n}}{\bar{s}u}},\bar{s}^{\bar{k}-1}\}$, where $\bar{s}\!=\!\bar{d}\!-\!\bar{k}\!+\!1$.

The authors in \cite{Chen} also developed the first explicit constructions of MSRR codes for all admissible parameters, i.e., $n\!=\!\bar{n}u, ~k\!=\!\bar{k}u+u_{0}~(0\!\leq \!u_{0}\!<\!u)$ and $\bar{k}\!\leq\!\bar{d}\!\leq\!\bar{n}\!-\!1$ \footnote{These parameters coincide with the assumptions made when proving the cut-set bound and deriving the MSRR code parameters in \cite{Hou}. Thus in this paper we regard this range as all admissible parameters for MSRR codes. }. However, their codes have sub-packetization $\bar{s}^{\bar{n}}$, higher than the proved lower bound. To our knowledge, no MSRR codes attaining the bounds on sub-packetization have been derived so far, even for the codes without the optimal-access property.

\subsection{Contribution and related work}
In this paper, we present an improved explicit construction of MSRR codes for all admissible parameters. Our code has sub-packetization $\bar{s}^{\lceil\frac{\bar{n}}{u-u_0}\rceil}$, thus taking a step towards shrinking the gap
between realization and proved lower bound. Moreover, we also reduce the field size almost by half. Namely, in \cite{Chen} the codes are built over a finite field $F$ satisfying $n\mid (|F|-1)$ and $|F|>n+\bar{s}-1$, which results in $|F|\geq 2n+1$, while our code needs $u\mid (|F|-1)$ and $|F|>n$ which results in $|F|\approx n$.

In \cite{Hou}, after derivation of the parameters for MSRR codes, the authors also discussed the construction.
They designed specific structure for satisfying the optimal repair while leaving the MDS property to
the Schwartz-Zippel Lemma. As a result, their constructions need some  constraints on the parameters and the finite fields being large enough.

The first explicit constructions of MSRR codes for all admissible parameters were developed in \cite{Chen}. Actually, two constructions were derived where both have the same sub-packetization level but the latter possesses lower access and smaller field size. Thus we only list the parameters of the second construction in \cite{Chen} for comparison in Table \ref{t0}. Note that our code keeps the same access level as their low-access construction, i.e., $\frac{u\alpha}{\bar{s}}$ symbols from each helper rack. Although it is by a factor of $u$ greater than the lower bound proved in \cite{Chen}, it is the lowest access among all existing constructions that are applicable to all admissible parameters.

In a recent work \cite{Hou2020}, Hou et al. present a coding framework for converting any $(\bar{n},\bar{k},\bar{d})$ MSR code into an $(n=\bar{n}u,k=\bar{k}u\!+\!u_0,\bar{d})$ MSRR code with the same sub-packetization. However, for arbitrary $\bar{n}$ and $\bar{k}$ all existing explicit constructions of $(\bar{n},\bar{k},\bar{d})$ MSR codes have sub-packetization $\bar{s}^{\bar{n}}$ except the ones in \cite{Ye2016sub-,Tang} that have sub-packetization $\bar{s}^{\lceil\frac{\bar{n}}{\bar{s}}\rceil}$ but only apply to $\bar{d}\!=\!\bar{n}\!-\!1$. By using the conversion framework, an $(n,k,\bar{d}\!=\!\bar{n}\!-\!1)$ MSRR code is obtained. However, the conversion again relies on the Schwartz-Zippel Lemma, so the MSRR code exists provided the finite field is sufficiently large. Comparisons between our MSRR code and previous constructions are shown in Table \ref{t0}.

The remaining of the paper is organized as follows. Section II describes a repair framework for MSRR codes that is used in both Chen\&Barg's codes and the code in this work.  Then Section III presents the explicit construction of MSRR codes.  Section IV concludes the paper.

\section{A Repair Framework for MSRR Codes}
First introduce some notations. For integers $0\leq m<n$, let $[n]=\{1,...,n\}$ and $[m,n]=\{m,m\!+\!1,...,n\}$. We label the racks from $0$ to $\bar{n}-1$ and the nodes within each rack from $0$ to $u-1$. Moreover, we represent each of the $n=\bar{n}u$ nodes by a pair $(e,g)\in[0,\bar{n}\!-\!1]\times [0,u\!-\!1]$ where $e$ is the rack index and $g$ is the node index within the rack.

In this section, we formalize the construction of MSRR codes from the parity check equations. Denote $r\!=\!n\!-\!k$ and $\bar{r}\!=\!\bar{n}\!-\!\bar{k}$ throughout the paper. Since the MSRR code is first an $(n,k;\alpha)$ MDS array code, the code can be defined by the following parity check equations.
\begin{equation}\label{PCE}
  \textstyle{\sum_{e=0}^{\bar{n}-1}\sum_{g=0}^{u-1}H_{(e,g)}{\bm c}_{(e,g)}^\tau}={\bm 0}\;,
\end{equation}
where $H_{(e,g)}$ is a $r\alpha\!\times\! \alpha$ matrix over a finite field $F$ and ${\bm c}_{(e,g)}\!=\!(c_{(e,g),0},...,c_{(e,g),\alpha-1})\!\in\! F^{\alpha}$ denotes the vector stored in node $(e,g)$. The MDS property means any $k$ out of the ${\bm c}_{(e,g)}$'s can recover all other $r$ vectors, which is equivalent to require the concatenation of any $r$ distinct $H_{(e,g)}$'s results in a $r\alpha\times r\alpha$ invertible matrix.

Besides, the MSRR codes should satisfy the optimal repair property. That is, each vector ${\bm c}_{(e^*,g^*)}\in F^{\alpha}$ can be recovered from
$\{{\bm c}_{(e^*,g)}\mid g\in[0, u-1], g\neq g^*\}\cup\{{\bm s}_e\mid e\in\mathcal{H}\}$
for any $\mathcal{H}\!\subseteq\! [0,\bar{n}\!-\!1]\!-\!\{e^*\}$ with $|\mathcal{H}|\!=\!\bar{d}$, where ${\bm s}_e\!\in F^{{\alpha}/({\bar{d}-\bar{k}+1})}$ is computed from $\{{\bm c}_{(e,g)}\mid g\in[0,u-1]\}$. The next theorem gives a sufficient condition for the optimal repair property.

\begin{theorem}\label{thm2}
Suppose $\mathcal{C}$ is an $(n,k;\alpha)$ array code defined by the parity check equations  in \eqref{PCE}. Denote $\beta=\alpha/(\bar{d}-\bar{k}+1)$. Then $\mathcal{C}$ satisfies the optimal repair property if for any $e^*\in[0,\bar{n}-1]$, there exists a matrix $S_{e^*}\in F^{\bar{r}\beta\times r\alpha}$ such that
\begin{itemize}
\item[(a)]For $g\!\in\![0,u\!-\!1]$, $S_{e^*}H_{(e^*,g)}\!=\!P_{e^*}Q_{(e^*,g)}$, where $Q_{(e^*,g)}$ is an $\alpha\times \alpha$ invertible matrix and $P_{e^*}\in F^{\bar{r}\beta\times \alpha}$;
\item[(b)]For all $e\!\neq\! e^*$ and $g\!\in\![0,u\!-\!1]$, $S_{e^*}H_{(e,g)}\!=\!P_eR_eQ_{(e,g)}$, where $P_e\!\in\! F^{\bar{r}\beta\times \beta}, R_e\!\in \!F^{\beta\times \alpha}, Q_{(e,g)}\!\in\! F^{\alpha\times\alpha}$.
\item[(c)]For any $\{e_1,...,e_{\bar{n}-\bar{d}-1}\}\!\in\! [0,\bar{n}\!-\!1]\!-\!\{e^*\}$, the matrix $\begin{pmatrix}P_{e^*}&P_{e_1}&\cdots&P_{e_{\bar{n}-\bar{d}-1}}\end{pmatrix}\in F^{\bar{r}\beta\times\bar{r}\beta}$ is invertible.
\end{itemize}
\end{theorem}
\begin{proof}
For any $e^*\!\in\![0,\bar{n}\!-\!1]$, we prove that existence of the matrix $S_{e^*}$ implies the optimal repair of any individual node in rack $e^*$. Actually, multiply $S_{e^*}$ from the left on both sides of \eqref{PCE}, then we have
\begin{equation}\label{eq3}
  P_{e^*}\sum_{g=0}^{u-1}Q_{(e^*,g)}{\bm c}_{(e^*,g)}^\tau+\sum_{e\neq e^*}P_eR_e\sum_{g=0}^{u-1}Q_{(e,g)}{\bm c}_{(e,g)}^\tau=\bm{0}\;.
\end{equation}
Furthermore, for all $e\in[0,\bar{n}-1]$ denote
\begin{equation}\label{eq40}
\tilde{\bm c}_e^\tau=\textstyle{\sum_{g=0}^{u-1}Q_{(e,g)}{\bm c}_{(e,g)}^\tau}\end{equation}
 then \eqref{eq3} becomes
\begin{equation}\label{eq50}
P_{e^*}\tilde{\bm c}_{e^*}^\tau+\textstyle{\sum_{e\neq e^*}P_e(R_e\tilde{\bm c}_e^\tau)}=\bm{0}\;.
\end{equation}
The condition (c) of the hypothesis implies that by downloading the vector ${\bm s}_{e}^\tau=R_{e}\tilde{\bm c}_{e}^\tau$ from the helper rack $e\in [0,\bar{n}-1]-\{e^*,e_1,...,e_{\bar{n}-1-\bar{d}}\}$, one can recover
$\big\{\tilde{\bm c}_{e^*}^\tau\big\}\cup\big\{R_{e_i}\tilde{\bm c}_{e_i}^\tau\!\mid\! i\!\in\![\bar{n}\!-\!\bar{d}\!-\!1]\}$.
Obviously, ${\bm s}_{e}\!\in\! F^\beta$, thus only $\beta$ symbols are downloaded from each helper rack. Moreover, from the condition (a) of the hypothesis one can further derive ${\bm c}_{(e^*,g^*)}$ from $\tilde{\bm c}_{e^*}$ and $\{{\bm c}_{(e^*,g)}\mid g\in[0,u-1], g\neq g^*\}$.
\end{proof}

\begin{remark}\label{rmk1}
 Theorem \ref{thm2} presents a specific but simpler repair framework for MSRR codes. More details are given below.
 \begin{enumerate}
 \item The matrix $S_{e^*}$ actually means selecting $\bar{r}\beta$ parity check equations from (\ref{PCE}) which then define an $(\bar{r}+\bar{d},\bar{d};\beta)$ MDS array code as shown in (\ref{eq50}), where for $e\neq e^*$, $R_e\tilde{\bm c}_e^\tau\in F^\beta$ represents one component of the MDS array codeword, and $\tilde{\bm c}_{e^*}^\tau\!\in\! F^{\alpha}$ represents $\bar{d}\!-\!\bar{k}\!+\!1$ components. The MDS property comes from the condition (c).
 \item The condition (a) and (b) guarantee that after multiplying the matrix $S_{e^*}$ a common divisor $P_e$ can be drawn out for each rack $e$. Therefore, all $u$ nodes in rack $e$ play as a whole (i.e., the $\tilde{\bm c}_e$ defined in (\ref{eq40})) in the repair process.
 \item The matrix $R_e$ means a compression from $\alpha$ symbols to $\beta$ symbols, while for the host rack $e^*$ there is no compression. This guarantees the ratio of downloaded data size to recovered data size.
 \item The condition (a) requires that $Q_{(e^*,g)}$, $g\!\in\![0,u\!-\!1]$, are invertible matrices, which implies the same selection of parity check equations (i.e., $S_{e^*}$) can be used for the repair of any single node failure in rack $e^*$.
 \end{enumerate}
Although Theorem \ref{thm2} proposes a stronger requirement than the optimal repair property, it also simplifies the design of MSRR codes and provides some insights into the constructions of \cite{Chen} and this work.
\end{remark}

\begin{remark}\label{rmk2}
The repair of single node failures in rack $e^*$ uses only part of the $r\alpha$ parity check equations in (\ref{PCE}) which exactly correspond to the nonzero columns of $S_{e^*}$. Divide the $r\alpha$ parity check equations into $r$ blocks each containing $\alpha$ equations. In \cite{Chen} a total of $\bar{r}$ blocks of check equations are used for the repair of single node failures in one rack.  By contrast, we use $\bar{r}(u-u_0)$ blocks of check equations to repair single node failures in $u-u_0$ racks. That is, more parity check equations are used to repair more racks in our construction. As a result, a smaller exponent (i.e., $\lceil\frac{\bar{n}}{u-u_0}\rceil$) in the sub-packetization is enough to ensure the repair of all $\bar{n}$ racks.
\end{remark}

\section{The Explicit Construction}
Suppose $k\!=\!\bar{k}u\!+\!u_{0}\ (0\!\leq \!u_{0}\!<\!u)$ and $\bar{k}\!\leq\!\bar{d}\!\leq\!\bar{n}-1$.
We construct an $(\bar{n}u,k,\bar{d})$ MSRR code $\mathcal{C}$ with sub-packetization $\alpha=\bar{s}^m$, where $\bar{s}=\bar{d}-\bar{k}+1$ and $m=\lceil\frac{\bar{n}}{u-u_{0}}\rceil$.
The code $\mathcal{C}$ is defined by parity check equations as in \eqref{PCE}. First we introduce some notations related to the expression of $H_{(e,g)}$'s.

\begin{itemize}
  \item Divide $H_{(e,g)}$ into $r$ row blocks
$H_{t,(e,g)}$, $t\in[0,r\!-\!1]$,
where $H_{t, (e,g)}\in F^{\alpha\times \alpha}$ is the $(t+1)$-th $\alpha$ rows of $H_{(e,g)}$.
\item Label the rows and columns of $H_{t, (e,g)}$ by the integers in $[0,\alpha-1]$. For any $a,b\in[0,\alpha-1]$, $H_{t,(e,g)}(a,b)$ denotes the $(a,b)$-th entry of $H_{t,(e,g)}$.
\item For each integer $a\!\in\![0,\alpha\!-\!1]$, let $(a_0,...,a_{m-1})$ be its $\bar{s}$-ary expansion, i.e., $a\!=\!\sum_{\tau=0}^{m-1}a_{\tau}\bar{s}^{\tau}$, $a_{\tau}\!\in\![0,\bar{s}\!-\!1]$. For any $v\!\in\![0,\bar{s}\!-\!1]$ and $\tau\!\in\![0,m\!-\!1]$, let $a(\tau,v)$ be the integer that has the $\bar{s}$-ary expansion $(a_{0},...,a_{\tau-1}, v, a_{\tau+1},...,a_{m-1})$.
\item For $e\!\in\![0,\bar{n}\!-\!1]$, define $\pi(e)\!=\!e\!-\!(u-u_0)\lfloor\frac{e}{u-u_0}\rfloor$, i.e., $e\equiv \pi(e)~{\rm mod~}(u-u_0)$.
\end{itemize}

Secondly we choose some specific elements in a finite field $F$, where $u|(|F|-1)$ and $|F|>n$.
\begin{enumerate}
  \item Let $\xi$ be a primitive element of $F$ and $\eta$ be an element of $F$ with multiplicative order $u$.
  \item Denote $\lambda_{(e,g)}=\xi^e\eta^g$ for $e\!\in\![0,\bar{n}\!-\!1],g\!\in\![0,u\!-\!1]$. It can be seen $\lambda_{(e,g)}\!\neq\!\lambda_{(e',g')}$ for $(e,g)\!\neq\! (e',g')\!\in\![0,\bar{n}\!-\!1]\times[0,u\!-\!1]$, because $(\xi^{e-e'})^u\neq 1$ for $e\!\neq\! e'\in[0,\bar{n}-1]$ while $(\eta^{g'\!-g})^u=1$ for all $g,g'\in[0,u-1]$.
  \item Let $\mu_1,\cdots,\mu_{\bar{s}-1}$ be $\bar{s}-1$ distinct nonzero elements in $F$ such that $\{\mu_{1},\cdots,\mu_{\bar{s}-1}\}\cap\{\xi^{eu}:e\in[0,\bar{n}-1]\}=\emptyset$. Note $\bar{s}-1+\bar{n}=\bar{d}-\bar{k}+\bar{n}<2\bar{n}$, so these $\mu_i$'s exist for $u\geq 2$ and $|F|>n$.
\end{enumerate}

Next we give Algorithm \ref{alg1} for defining the $H_{t,(e,g)}$'s. The whole parity check matrix is established by running Algorithm \ref{alg1} for $t\in[0,r-1]$.


\begin{algorithm}[th]
\caption{\\Defining $H_{t,(e,g)}$'s for $e\in[0,\bar{n}-1]$ and $g\in[0,u-1]$.}\label{alg1}
\begin{algorithmic}[1]
\STATE Diagonal: for $a\!\in\![0,\alpha\!-\!1]$,  set $H_{t,(e,g)}(a,a)\!=\!\lambda_{(e,g)}^{t}$;
\STATE Non-diagonal:
\FOR{$e\in[0,\bar{n}-1]$, $g\in[0,u-1]$ and $a\!\in\![0,\alpha\!-\!1]$}
\STATE Initialize $H_{t,(e,g)}(a,b)\!=0$ for all $b\neq a$;
\STATE Denote $\tau=\lfloor\frac{e}{u-u_0}\rfloor$;
~~\IF{$a_{\tau}=0$ and $t\equiv \pi(e)~{\rm mod~}u$}
\STATE Set $H_{\!t,(e,g)\!}(a,b)\!=\!\lambda_{(e,g)}^{\pi(e)}\mu_{v}^{\lfloor\frac{t}{u}\rfloor}$ for $b\!\!=\!\!a(\tau,v),v\!\in\![\bar{s}\!-\!1]$;
\ENDIF
\ENDFOR
\end{algorithmic}
\end{algorithm}
We give some explanations of Algorithm \ref{alg1}. Actually, Line 1 defines the diagonal entries of $H_{t,(e,g)}$'s, Line 4 initializes all non-diagonal entries as zeros, and then Line 6-7 updates the non-diagonal entries in some blocks (i.e., $t\equiv \pi(e)~{\rm mod~}u$), some rows (i.e., $a_{\tau}=0$) and some columns (i.e, $b\!\in\!\{a(\tau,v)\!\mid\! v\!\neq\! 0\}$).
In the following we prove $\mathcal{C}$ is an MSRR code by showing it satisfies the MDS property and optimal repair property.

\begin{remark}\label{re1}
The proofs are derived in an inductive way, which depends on a partition on the coordinates of a vector in $F^\alpha$. In more detail, for each vector in $F^\alpha$, its coordinates are indexed by subscripts ranging in $[0,\alpha\!-\!1]$.
For any $a\!\in\![0,\alpha\!-\!1]$, let $w(a)$ be the number of digits that equal $0$ in $a$'s $\bar{s}$-ary expansion $(a_{0},...,a_{m-1})$. Denote $\mathcal{L}_{\sigma}\!=\!\{a\!\in\![0,\alpha\!-\!1]\!\mid\!\omega(a)\!=\!\sigma\}$. Obviously, $\cup_{\sigma=0}^m\mathcal{L}_{\sigma}$ forms a partition of the set $[0,\alpha\!-\!1]$. We prove the two properties of $\mathcal{C}$ by induction on $\sigma$.
\end{remark}

\subsection{Proof of the MDS property}
\begin{theorem}\label{thm3}
The code $\mathcal{C}$ satisfies the MDS property, i.e., for any $r$ nodes $(e_1,g_1),...,(e_r,g_r)\in[0,\bar{n}-1]\times[0,u-1]$, the matrix $H=(H_{(e_1,g_1)}~H_{(e_2,g_2)}~\cdots~H_{(e_r,g_r)})$ is invertible.
\end{theorem}

\begin{proof}
It suffices to show for any $\bm x\!\in\!(F^{\alpha})^r$, $H{\bm x} ^{\tau}\!=\bm 0$ always implies $\bm x\!=\!\bm 0$. Denote $\bm x\!=\!(\bm x_1,...,\bm x_r)$ and $\bm x_{i}\!=\!(x_{i,0},x_{i,1},...,x_{i,\alpha-1})\in F^{\alpha}$ for $i\!\in\! [r]$. Using the partition defined in Remark \ref{re1}, next we prove $\bm x=0$ by showing $\{x_{i,\mathcal{L}_{\sigma}}\!\mid \!i\!\in\![r]\}$ contains only zeros for all $\sigma\in [0,m]$. This is accomplished by induction on $\sigma$.

For simplicity, denote $H_t=(H_{t,(e_1,g_1)}~\cdots~H_{t,(e_r,g_r)})$
for $t\in[0,r-1]$. Then the linear system $H{\bm x} ^{\tau}\!=\bm 0$ becomes
\begin{equation}\label{eq4}
H_t{\bm x}^\tau=\textstyle{\sum_{i=1}^rH_{t,(e_i,g_i)}{\bm x}_i^\tau}={\bm 0},~\forall t\in[0,r-1]\;.
\end{equation}

First consider the base case $\sigma=0$.
For any $a\in \mathcal{L}_{0}$, by the definition of $H_{t,(e,g)}$ in Algorithm \ref{alg1} we know the $a$-th row of $H_{t,(e,g)}$ are all zeros except the $(a,a)$-th entry.
Choose the $a$-th rows in the linear system \eqref{eq4}, one can obtain the following linear system
\begin{equation}\label{B-1}
\textstyle{\sum_{i=1}^r\lambda_{(e_i,g_i)}^{t}x_{i,a}}=0,~\forall t\in[0,r-1].
\end{equation}
Since $\lambda_{(e_1,g_1)},...,\lambda_{(e_r,g_r)}$ are distinct elements in $F$, it immediately follows $x_{1,a}=\cdots=x_{r,a}=0$. Thus $\{x_{i,\mathcal{L}_{0}}\mid i\in[r]\}$ contains only zeros.
	
Now suppose it has been proved $\{x_{i,\mathcal{L}_{\sigma}}\!\mid\! i\in[r]\}$ contains only zeros for some $\sigma\geq 0$. Then for any $a\in \mathcal{L}_{\sigma+1}$, the $a$-th rows in \eqref{eq4} are
\begin{equation}\label{B-2}
\sum_{i=1}^r\lambda_{(e_i,g_i)}^{t}x_{i,a}+\sum_{i=1}^r\sum_{v=1}^{\bar{s}-1}f_{t}(a,e_i)
\lambda_{(e_i,g_i)}^{\pi(e_i)}\mu_{v}^{\lfloor\frac{t}{u}\rfloor}x_{i,a(\lfloor\frac{e_i}{u-u_0}\rfloor,v)}=0,
\end{equation}
where
$$
f_{t}(a,e_i)=\begin{cases}
1\ \ \ \ \ \ \ \mathrm{if} \ a_{\lfloor\frac{e_i}{u-u_0}\rfloor}=0\mathrm{~and~} t\equiv \pi(e_i)~{\rm mod~}u\\
0\ \ \ \ \ \ \ \mathrm{otherwise}.
\end{cases}
$$
However, for the parameters $t,e_i$ such that $f_{t}(a,e_i)\!\neq\!0$, it must have $a(\lfloor\frac{e_i}{u-u_0}\rfloor,v)\!\in\!\mathcal{L}_{\sigma}$ for $v\!\in\![\bar{s}\!-\!1]$, and then $x_{i,a(\lfloor\frac{e_i}{u-u_0}\rfloor,v)}=0$ by the induction hypothesis. As a result, \eqref{B-2} becomes $\sum_{i=1}^r\!\lambda_{(e_i,g_i)}^{t}x_{i,a}\!=\!0$ for $t\!\in\![0,r\!-\!1]$. Similar to \eqref{B-1}, it follows $x_{1,a}=\cdots=x_{r,a}=0$. Thus $\{x_{i,\mathcal{L}_{\sigma+1}}\!\mid\!i\in[r]\}$ contains only zeros. Therefore, the inductive proof is finished.
\end{proof}

\subsection{Proof of the repair property}
\begin{theorem}\label{thm4}
The code $\mathcal{C}$ satisfies the optimal repair property, i.e., for any node $(e^*,g^*)$ and any $\mathcal{H}\!\subseteq\! [0,\bar{n}\!-\!1]\!-\!\{e^*\}$ with $|\mathcal{H}|\!=\!\bar{d}$, the vector ${\bm c}_{(e^*,g^*)}$ can be recovered from
$$\{{\bm c}_{(e^*,g)}\mid g\in[0,u-1],g\neq g^*\}\cup\{{\bm s}_e\mid e\in\mathcal{H}\}$$
where ${\bm s}_e\!\in F^{\beta}$ is computed from $\{{\bm c}_{(e,g)}\mid g\in[0,u-1]\}$.
\end{theorem}
\begin{proof}
We firstly select a system of the parity check equations with respect to the values of $t$, i.e.,
\begin{equation}\label{eq8}
\sum_{e=0}^{\bar{n}-1}\sum_{g=0}^{u-1}H_{t,(e,g)}{\bm c}_{(e,g)}^\tau={\bm 0}, ~\forall ~t\in T_{e^*}
\end{equation}
where $T_{e^*}\!=\!\{t\in[0,r-1]\!\mid \!t\!\equiv\! \pi(e^*)~{\rm mod~}u\}$. Since $r=n-k=(\bar{n}-\bar{k})u-u_0=\bar{r}u-u_0$, it obviously has
\begin{equation}\label{eqt}
T_{e^*}=\{\pi(e^*)+iu\mid i\in[0,\bar{r}-1]\}\;.\end{equation}
Denote $\tau^*\!=\!\lfloor\frac{e^*}{u-u_0}\rfloor$ and $A(\tau^*,0)\!=\!\{a\!\in\![0,\alpha\!-\!1]\mid a_{\tau^*}=0\}$.
Then, for all $a\!\in\! A(\tau^*,0)$ we pick the $a$-th rows from the equations in \eqref{eq8} which will be used to enable the repair of single node failures in rack $e^*$.

For simplicity, denote
$\mathcal{A}_\sigma=A(\tau^*,0)\cap\mathcal{L}_\sigma$ for $\sigma\in[m]$. Obviously, $\cup_{\sigma=1}^m\mathcal{A}_\sigma$ forms a partition of $A(\tau^*,0)$.
First consider the $a$-th rows in \eqref{eq8} for all $a\in\mathcal{A}_1$ which
induce the following linear system
\begin{align}
&\sum_{g=0}^{u-1}\lambda^{iu}_{(e^*,g)}\cdot\lambda_{(e^*,g)}^{\pi(e^*)}{ c}_{(e^*,g),a}\!+
\!\sum_{v=1}^{\bar{s}-1}\mu_v^i\!\Big(\!\sum_{g=0}^{u-1}\lambda_{(e^*,g)}^{\pi(e^*)}{ c}_{(e^*,g),a(\tau^*,v)}\!\Big)\notag\\
&+\!\sum_{e\neq e^*}\sum_{g=0}^{u-1}\lambda^{iu}_{(e,g)}\cdot\lambda_{(e,g)}^{\pi(e^*)}{ c}_{(e,g),a}\!=\!0,~~~\forall~ i\!\in\![0,\bar{r}\!-\!1]\;. \label{eq9}
\end{align}

We give some explanations about \eqref{eq9}. By Algorithm \ref{alg1}, for any $a\!\in\!\mathcal{A}_1$ and $t\!\in\! T_{e^*}$ the $a$-th row of $H_{t,(e^*\!,g)}$ has nonzero entries in the diagonal position and $\bar{s}\!-\!1$ non-diagonal positions, which respectively correspond to the first two terms in the left side of \eqref{eq9}. For any $e\!\neq\! e^*$, it has $\big(\lfloor\frac{e}{u-u_0}\rfloor,\pi(e)\big)\neq \big(\tau^*,\pi(e^*)\big)$. Combining with the fact that $a_\tau\!\neq\! 0$ for all $\tau\!\neq\!\tau^*$ due to $a\!\in\!\mathcal{A}_1$, the conditions $a_{\lfloor\frac{e}{u-u_0}\rfloor}=0$ and $t\equiv \pi(e)~{\rm mod~}u$ can not simultaneously hold for all $t\!\in\! T_{e^*}$. Therefore, the $a$-th rows of $H_{t,(e,g)}$'s only have nonzero entries in the diagonal positions which result in the third term in the left side of \eqref{eq9}. Moreover, according to the expression of $T_{e^*}$ in \eqref{eqt}, one can finally derive \eqref{eq9}.

Then for all $e\in[0,\bar{n}-1]$, denote
\begin{equation}\label{eq7}
\tilde{\bm c}_{e}=\sum_{g=0}^{u-1}\lambda_{(e,g)}^{\pi(e^*)}{\bm c}_{(e,g)}=(\tilde{ c}_{e,0},...,\tilde{ c}_{e,\alpha-1})\in F^\alpha\;.
\end{equation}
Obviously, ${\bm c}_{(e^*,g^*)}$ can be computed from $\tilde{\bm c}_{e^*}$ and the intra-rack transmission $\{{\bm c}_{(e^*,g)}\!\mid\! g\!\in\![0, u-1], g\!\neq\! g^*\}$.

Moreover, because $\lambda_{(e,g)}=\xi^e\eta^g$ and $\eta$ has multiplicative order $u$, it has $\lambda_{(e,g)}^{iu}=(\xi^{eu})^i$. Using the notation defined in \eqref{eq7}, the linear system \eqref{eq9} becomes
\begin{equation}\label{eq11}
\sum_{e=0} ^{\bar{n}-1}(\xi^{eu})^i\tilde{c}_{e,a}\!+\!\sum_{v=1}^{\bar{s}-1}\!\mu_v^i\tilde{c}_{e^*\!,a(\tau^*\!,v)}\!=\!0,~\forall i\!\in\![0,\bar{r}\!-\!1]\;.
\end{equation}
By the selection of $\xi$ and $\mu_v$'s, \eqref{eq11} actually defines a $(\bar{r}\!+\!\bar{d},\bar{r})$ GRS codeword $(\tilde{c}_{0,a},...,\tilde{c}_{\bar{n}-1,a},\tilde{c}_{e^*\!,a(\tau^*\!,1)},...,
\tilde{c}_{e^*\!,a(\tau^*\!,\bar{s}-1)})$, so downloading $\{\tilde{c}_{e,a}\mid e\in\mathcal{H}\}$ can recover $\{\tilde{c}_{e^*\!,a},\tilde{c}_{e^*\!,a(\tau^*\!,1)},...,\\\tilde{c}_{e^*\!,a(\tau^*\!,\bar{s}-1)}\}\cup
\{\tilde{c}_{e,a}\mid e\in[0,\bar{n}\!-\!1]\!-\!\mathcal{H}\}$.

Furthermore, we prove
$\{\tilde{c}_{e^*\!,a(\tau^*\!,0)},...,
\tilde{c}_{e^*\!,a(\tau^*\!,\bar{s}-1)}\!\mid \!a\!\in\! \mathcal{A}_{\sigma}\}$ can be recovered from $\{\tilde{c}_{e,b}\!\mid\! b\!\in\!  \cup_{\delta=1}^\sigma\mathcal{A}_{\delta}, e\!\in\!\mathcal{H}\}$ for all $\sigma\in[m]$. This is accomplished by induction on $\sigma$ and the above is the proof for the base case $\sigma\!=\!1$.

Let us see the inductive step.  For any $a\in \mathcal{A}_{\sigma+1}$, we still pick the $a$-th rows from the parity check equations in \eqref{eq8} and obtain a linear system similar to \eqref{eq11} except the left side has the fourth term corresponding to the nonzero non-diagonal entries in the $a$-th rows of $H_{t,(e,g)}$ for the $e\neq e^*$ satisfying $\pi(e)=\pi(e^*)$ and $a_{\lfloor\frac{e}{u-u_0}\rfloor}=0$.
However, from $e\neq e^*$ and $\pi(e)=\pi(e^*)$, it must have $\lfloor\frac{e}{u-u_0}\rfloor\neq\lfloor\frac{e^*}{u-u_0}\rfloor=\tau^*$, thus $a(\lfloor\frac{e}{u-u_0}\rfloor,v)\in \mathcal{A}_{\sigma}$ for $v\in[\bar{s}-1]$. Therefore, by the induction hypothesis the fourth term can be computed from $\{\tilde{c}_{e,b}\mid e\in\mathcal{H}, b\!\in\! \cup_{\delta=1}^\sigma\mathcal{A}_{\delta}\}$. Then similar to \eqref{eq11}, one can recover $\{\tilde{c}_{e^*\!,a},\tilde{c}_{e^*\!,a(\tau^*\!,1)},...,\tilde{c}_{e^*\!,a(\tau^*\!,\bar{s}-1)}\}\bigcup
\{\tilde{c}_{e,a}\mid e\in[0,\bar{n}-1]-\mathcal{H}\}$ by additionally downloading $\{\tilde{c}_{e,a}\mid e\in\mathcal{H}\}$ for all $a\in\mathcal{A}_{\sigma+1}$.

Therefore,  by downloading ${\bm s}_e\!=\!(\tilde{c}_{e,a})_{a\in A(\tau^*,0)}\!\in\! F^{\beta}$ from each helper rack $e\!\in\!\mathcal{H}$ along with the intra-rack communication, the repair is accomplished.
\end{proof}

\begin{remark}\label{Re2}
Although Theorem \ref{thm4} is proved by an inductive process according to a partition of the coordinates (see Remark \ref{re1}), it actually coincides with the sufficient conditions given in Theorem \ref{thm2} for the optimal repair.
\begin{enumerate}
  \item Selection of parity check equations for repair. In Theorem \ref{thm2} the matrix $S_{e^*}$ selects a linear system from (\ref{PCE}) which then induces an MDS array code defined in (\ref{eq50}). In Theorem \ref{thm4} this selection is sequentially accomplished by the restriction to the set $T_{e^*}$ defined in (\ref{eqt}) and then to the rows indexed by $A(\tau^*,0)$. The resultant MDS code is defined in (\ref{eq11}).

      Since $T_e\!=\!T_{\pi(e)}$ for all $e\!\in\![0,\bar{n}\!-\!1]$ and $\pi(e)$ ranges in $[0,u\!-\!u_0\!-\!1]$, $u-u_0$ linear systems are used for repair in our code. By contrast, \cite{Chen} only used the linear system labeled by $T_0$ for repair.
  \item All $u$ nodes in a rack play as a whole in the repair. From (\ref{eq7}) one can see our code $\mathcal{C}$ also follows this rule. Specifically, since in (\ref{eq9}) it has $\lambda_{(e,g)}^{iu}=(\xi^{eu})^i$ for all $g\in[0,u-1]$, $(\xi^{eu})^i$ is like the common divisor $P_e$ drawn out for each rack $e$ in Theorem \ref{thm2}, and the diagonal matrix $\lambda_{(e,g)}^{\pi(e^*)}I_{\alpha}$ corresponds to the matrix $Q_{(e,g)}$ in Theorem \ref{thm2}, where $I_\alpha$ is the $\alpha\times\alpha$ identity matrix. Obviously, $\lambda_{(e,g)}^{\pi(e^*)}I_{\alpha}$ is invertible for all $e\in[0,\bar{n}-1]$.
\end{enumerate}

\end{remark}

\begin{remark}
From \eqref{eq7} and the proof of Theorem \ref{thm4} one can easily see the repair of a node failure needs to access $\alpha/\bar{s}$ symbols from each node in the helper racks, which is the same as the low-access construction in \cite{Chen}.
\end{remark}

\section{Conclusion and Future Work}
In this work, by using the parity-check equations in an more efficient way for repair, we reduce the sub-packetization of existing explicit constructions of MSRR codes from $(\bar{d}-\bar{k}+1)^{\bar{n}}$ to $(\bar{d}-\bar{k}+1)^{\lceil\frac{\bar{n}}{u-u_{0}}\rceil}$, which helps to bridge the gap from the proved lower bound. Further reducing the sub-packetization and proving a lower bound without the optimal-access hypothesis are left as future work. Besides, constructing optimal-access MSRR codes for nontrivial parameters seems to be an even harder problem.


\begin{thebibliography}{100}

\bibitem{Dimakis2011}
A. G. Dimakis, P. B. Godfrey, Y. Wu, M. J. Wainwright, and K. Ramchandran, ``Network coding for distributed storage systems,'' IEEE Trans. Inform. Theory, vol. 56, no. 9, pp. 4539-4551, Sep. 2010.

\bibitem{Kumar2011}
K.~V. Rashmi, N.~B. Shah, and P.~V. Kumar, ``Optimal exact-regenerating codes for distributed storage at the MSR and MBR points via a product-matrix construction,'' IEEE Trans. Inform. Theory, vol. 57, no. 8, pp. 5227-5239, Aug. 2011.

\bibitem{Sasidharan2015}
B.~Sasidharan, G.~K. Agarwal, and P.~V. Kumar, ``A high-rate MSR code with polynomial sub-packetization level,'' IEEE International Symposium on Information Theory, Oct. 2015.

\bibitem{Rawat2016}
A. S. Rawat, O. O. Koyluoglu, and S. Vishwanath, ``Progress on High-rate MSR codes: Enabling Arbitrary Number of Helper nodes,'' Information Theory and Applications Workshop(ITA), Feb. 2016.

\bibitem{Ye2016}
M.~Ye, A.~Barg, ``Explicit constructions of high rate MDS array codes with optimal repair bandwidth,'' IEEE Trans. Inform. Theory, vol. 63, no. 4, pp. 2001-2014, Apr. 2017.
	
\bibitem{SubPBoundSTOC19}
O. Alrabiah, V. Guruswami, ``An exponential lower bound on the sub-packetization of MSR codes,'' STOC, pp. 979-985, 2019.

\bibitem{Ye2016sub-}
M.~Ye, A.~Barg, ``Explicit constructions of optimal-access MDS codes with nearly optimal sub-packetization,'' IEEE Trans. Inform. Theory, vol. 63, no. 10, pp. 6307-6317, Oct. 2017.

\bibitem{Wang}
Z.~Wang, I.~Tamo, J.~Bruck, ``Long MDS Codes for Optimal Repair Bandwidth,'' IEEE International Symposium on Information Theory Proceedings, pp. 1182-1186, July 2012.

\bibitem{Hu}
Y.~Hu, P.~P. C. Lee, and X.~Zhang, ``Double regenerating codes for hierarchical data centers,'' Proc IEEE Int. Sympos. Inform. Theory(ISIT), pp. 245-249, July 2016.

\bibitem{Hou}
H.~Hou, P.~Lee, K.~Shum, and Y.~Hu, ``Rack-aware regenerating codes for data centers,'' IEEE Trans. Inform Theory, vol. 65, no. 8, pp. 4730-4745, 2019.

\bibitem{Chen}
Z.~Chen, A.~Barg, ``Explicit constructions of MSR codes for clustered distributed storage: The rack-aware storage model,'' IEEE Trans. Inform Theory, vol. 66, no. 2, pp. 886-899, Feb. 2020.

\bibitem{Balaji2018}
S. Balaji, P. Vijay Kumar, ``A tight lower bound on the sub-packetization level of optimal-access MSR and MDS codes,'' IEEE International Symposium on Information Theory, pp. 2381-2385, June 2018.

\bibitem{Hou2020}
H.~Hou, P.~Lee, and Y.~Han, ``Minimum Storage Rack-Aware Regenerating Codes with Exact Repair and Small Sub-Packetization,'' IEEE International Symposium on Information Theory, pp. 554-559, June 2020.

\bibitem{Tang}
J. Li, X. Tang, and C. Tian, ``A Generic Transformation for Optimal
Repair Bandwidth and Rebuilding Access in MDS Codes,'' IEEE International Symposium on Information Theory, pp. 1623-1627, June 2017.

\end{thebibliography}
\end{document}